\documentclass[final,hidelinks,onefignum,onetabnum,a4paper]{siamart250211}

\usepackage{lipsum}
\usepackage{amsfonts}
\usepackage{graphicx}
\usepackage{epstopdf}
\usepackage[noend]{algpseudocode}
\ifpdf
  \DeclareGraphicsExtensions{.eps,.pdf,.png,.jpg}
\else
  \DeclareGraphicsExtensions{.eps}
\fi

\newsiamremark{remark}{Remark}
\newsiamremark{hypothesis}{Hypothesis}
\crefname{hypothesis}{Hypothesis}{Hypotheses}
\newsiamthm{claim}{Claim}
\newsiamthm{statement}{Statement}
\newsiamthm{convention}{Convention}
\newsiamremark{fact}{Fact}
\crefname{fact}{Fact}{Facts}

\title{On the two paths theorem \\ and the two disjoint paths problem}

\headers{On the two paths theorem and the two disjoint paths problem}{S. Humeau and D. Pous}
\author{Samuel Humeau$^\ast$ \and Damien Pous%
  \thanks{Plume, LIP, CNRS, ENS de Lyon, France \\
    \email{samuel.humeau@ens-lyon.fr}, \url{https://perso.ens-lyon.fr/samuel.humeau}\\
    \email{damien.pous@ens-lyon.fr}, \url{https://perso.ens-lyon.fr/damien.pous}}
}

\usepackage{amsopn}

\begin{document}

\maketitle

\begin{abstract}
  A tuple $(s_1,s_2,t_1,t_2)$ of vertices in an undirected graph
  is \emph{$2$-linked} when there are two vertex-disjoint paths,
  respectively from $s_1$ to $t_1$ and $s_2$ to $t_2$. A graph is
  $2$-linked when all such tuples are $2$-linked.

  We give a new and simple proof of the ``two paths theorem'', a
  characterisation of edge-maximal graphs which are not $2$-linked as
  \emph{webs}: particular near triangulations filled with cliques. Our
  proof works by generalising the theorem, replacing the four vertices
  above by an arbitrary tuple; it does not require major theorems such
  as Kuratowski's or Menger's theorems. Instead, it follows an
  inductive characterisation of generalised webs via \emph{parallel
    composition}, an operation consisting in taking a disjoint
  union and identifying some pairs of vertices.

  We use the insights provided by this proof to design a simple
  $O(nm)$ recursive algorithm for the ``two vertex-disjoint paths''
  problem. This algorithm returns either two disjoint paths, or an
  embedding of the input graph into a web.
\end{abstract}



\section{Introduction}
\label{sec:introduction}

\subsection{Two paths theorem}
\label{ssec:intro:thm}

A graph is called \emph{$2$-linked} if for all pairwise distinct
vertices $s_1,t_1,s_2,t_2$, there are two disjoint paths, respectively
from $s_1$ to $t_1$ and from $s_2$ to $t_2$. Such a pair of paths is
called an \emph{$(s_1,s_2,t_1,t_2)$-linkage}.

A \emph{rib} is a planar graph such that in one of its plane
embeddings, the outer face is a cycle of length four, all inner faces
are triangles, and all triangles are inner faces. A \emph{web} is a
graph constructed from a rib $R$ by adding one clique $K_T$ per
triangle $T$ of $R$, making vertices of $K_T$ adjacent to exactly the
vertices of~$T$. The outer face of a rib or of a web is called its
\emph{frame}.
\begin{theorem}[Two paths
  theorem~\cite{DBLP:journals/ejc/Thomassen80,DBLP:journals/dm/Seymour80}]
  \label{theo:main}
  The edge-maximal graphs which are not $2$-linked are the webs.
\end{theorem}

This theorem is easily shown equivalent to the following one,
where we call \emph{edge extension} of a graph any graph obtained
from it by adding a new edge.
\begin{theorem}[Two paths theorem variant]
  \label{theo:variant}
  Let $s_1,s_2,t_1,t_2$ be four distinct vertices in a graph $G$. The following are equivalent:
  \begin{enumerate}
  \item $G$ is a web with frame $s_1s_2t_1t_2s_1$,
  \item every edge extension of $G$ contains an $(s_1,s_2,t_1,t_2)$-linkage, but not $G$.
  \end{enumerate}
\end{theorem}

To prove this theorem, Thomassen~\cite{DBLP:journals/ejc/Thomassen80}
reduces the problem along small separators and then relies on
Kuratowski's theorem. Independently,
Seymour~\cite{DBLP:journals/dm/Seymour80} sketched a proof reducing
the problem to the $2$-connected case before doing a clever
combinatorial analysis of the graph's separators. Both proofs moreover
use Menger's theorem.

We give a new and direct proof of this theorem. We consider for that
the slight generalisation of ribs and webs where the frame (the outer
cycle in the near triangulation of the plane) does not need to have
length four, but some number $k\ge 3$.  Small webs are depicted in
Figure~\ref{fig:smallribs}. These are actually ribs; filling some of
their triangles with arbitrary cliques would give examples of webs
which are not ribs.
\begin{figure}[t]
  \centering
  \includegraphics{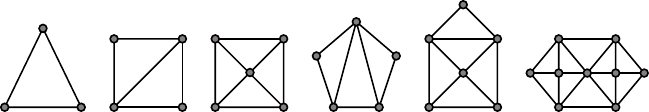}
  \caption{Small ribs, with frames of size three, four, four, five, five, and six.}
  \label{fig:smallribs}
\end{figure}

Accordingly, we also generalise the notion of linkage: a cycle
$C=x_1x_2\dots x_kx_1$ in a graph is \emph{crossed} when there are two
vertex-disjoint paths from $x_i$ to $x_j$ and $x_r$ to $x_s$ with
$i<r<j<s$, and \emph{crossless} otherwise. A graph $G$ is
\emph{maximally $C$-crossless} when $C$ is crossless in $G$, but
crossed in any edge extension of $G$.

Observe that none of the frames is crossed in the ribs from
Figure~\ref{fig:smallribs}, and that adding any edge to those graphs
makes their frame crossed. These facts remain unchanged under filling
triangles with cliques to obtain proper webs; this leads us to the
following generalisation of the two paths theorem:
\begin{theorem}[Generalised two paths theorem]
  \label{theo:gen}
  Let $G$ be a graph with a cycle~$C$. The following are equivalent:
  \begin{itemize}
  \item $G$ is a web with frame $C$,
  \item $G$ is maximally $C$-crossless.
  \end{itemize}
\end{theorem}

We define a simple operation on webs, which we call \emph{web
  composition}: two webs can be ``glued'' along induced subpaths of
their frames, resulting in a new web. We shall prove that all webs can
be constructed from cliques (with triangles as frames) by successive
web compositions. For instance, the ribs from
Figure~\ref{fig:smallribs} can be constructed as illustrated in
Figure~\ref{fig:smallribs:decomp}.
\begin{figure}[t]
  \centering
  \includegraphics{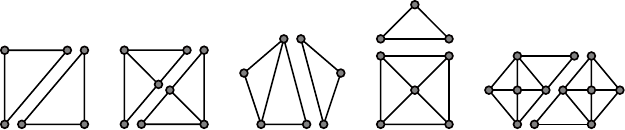}
  \caption{Decompositions of the non-trivial ribs from Figure~\ref{fig:smallribs}.}
  \label{fig:smallribs:decomp}
\end{figure}

Equipped with this tool, we proceed as follows to prove the difficult
implication of the theorem (maximally $C$-crossless graphs are webs
with frame $C$).  Suppose $C=x_1x_2y_1y_2x_1$ has size four, the
overall process is depicted in Figure~\ref{fig:process}. At first, we
only know that the graph contains $C$: edges in this cycle cannot be
used in a crossing. We define a path $P$ from $x_1$ to $y_1$,
initially composed of neighbours of $y_2$ only, which we transform to
make it closer and closer to $x_1x_2y_1$, progressively revealing the
web-structure of the starting graph.

\begin{figure}
  \centering
  \includegraphics[width=.8\linewidth]{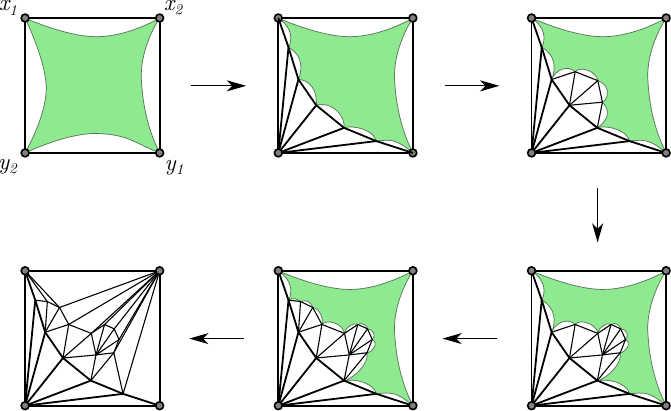}
  \caption{Recursively revealing the web-structure of a maximally crossless graph.}
  \label{fig:process}
\end{figure}

At each stage in Figure~\ref{fig:process}, the bottom-left part of the
graph is a web with frame $y_2Py_2$, and the top-right coloured one is
a maximally $x_2Px_2$-crossless graph. We process this latter part
recursively to obtain a web, and we compose it with the former along
the induced subpath $P$ to obtain the desired web.

\subsection{Two disjoint paths problem}
\label{ssec:intro:algo}

The two paths theorem is inherently related to the algorithmic problem
of finding vertex-disjoint paths between fixed pairs of vertices of a
given graph. This problem has been widely studied in the literature;
it is \textsc{NP}-complete on directed graphs~\cite{FORTUNE1980111},
and polynomial on undirected
graphs~\cite{DBLP:journals/dm/Seymour80,DBLP:journals/jacm/Shiloach80,DBLP:journals/ejc/Thomassen80}.

Many variations of this problem exist, such as optimisation
versions~\cite{https://doi.org/10.1002/net.3230140209}, acyclic
case~\cite{DBLP:journals/jacm/PerlS78}, etc. We focus on the following
variant: given vertices $s_1,t_1,s_2,t_2$ in a simple undirected
graph, find vertex-disjoint paths respectively from $s_1$ to $t_1$ and
from $s_2$ to $t_2$. If $s_1=s_2$ or $t_1=t_2$, then this problem is
solved by Menger's theorem: say $s_1=s_2$; the two paths exist if and
only if every vertex separator between $s_1$ and $\{t_1,t_2\}$ has at
least two elements. Various methods relating to flow networks exist to
solve such a problem. The difficulty arises when $s_1\neq s_2$ and
$t_1\neq t_2$. The corresponding algorithmic problem is usually called
the \emph{two disjoint paths problem}. Menger's theorem would apply to
find vertex-disjoint paths between $\{s_1,s_2\}$ and $\{t_1,t_2\}$,
but requiring paths from $s_1$ to $t_1$ and $s_2$ to $t_2$,
respectively, makes the problem harder.

Since the seminal works
in~\cite{DBLP:journals/dm/Seymour80,DBLP:journals/jacm/Shiloach80,DBLP:journals/ejc/Thomassen80},
several efficient algorithms have been introduced. For instance,
algorithms running in $O(n^2)$, $O(n+m\log n)$, and
$O(n+m\alpha(m,n))$ are respectively proposed in
\cite{DBLP:journals/siamcomp/KhullerMV92,gustedt,DBLP:journals/mst/Tholey06},
where $n$ and $m$ are the number of vertices and edges, respectively,
and $\alpha$ is the two parameters variant of the inverse of the
Ackermann function. All are based on improving the $O(nm)$ algorithm
by Shiloach~\cite{DBLP:journals/jacm/Shiloach80}. The latter relies on
a sequence of six reductions, three of which involve previous works:
instances can be reduced to the $3$-connected case (e.g. using the
algorithm from~\cite{doi:10.1137/0202012} to decompose a graph into
its triconnected components), and even the $3$-connected non-planar
case~\cite{DBLP:journals/jacm/PerlS78}, both in $O(n+m)$ time. By
Kuratowski's Theorem the resulting instance contains a topological
$K_{3,3}$ or $K_5$. Another reduction allows one to assume that the
instance does not contain a
$K_5$~\cite{10.1215/S0012-7094-68-03523-0}. To quote Shiloach ``We
have not written the solution as an explicit, long, and tedious
algorithm''.

We use the insights of our proof of the two paths theorem to design a
simple $O(nm)$ algorithm for the following slight generalisation of
the two disjoint paths problem. Given a tuple $T$ of at least three
vertices in a graph $G$, we call \emph{web completion} of $G$ with
regard to $T$ a set $F$ of edges such that $G+F$ is a web with frame
$T$.

\vspace{1cm}%
\noindent\textsc{Two Disjoint Paths or Web Completion}\\
\noindent\textbf{Input:} A graph with a tuple of distinct vertices.\\
\noindent\textbf{Output:} A crossing or a web completion.
\vspace{1cm}%

In contrast to Shiloach's method, no reductions are required: our
algorithm consists of two procedures of ten lines of pseudo-code each,
which are self-contained but for calls to some search algorithm
(e.g. depth-first or bread-first).

Its principle, in the decision version of the problem, goes as
follows: given a graph $G$ with a tuple of vertices $T$, if there are
no paths in $G$ between non-consecutive vertices of $T$, then return
No---$T$ is not crossed. Otherwise, compute a path $P$ between
non-consecutive vertices of $T$. If $P$ does not separate $T$, then
return Yes---$T$ is crossed. Otherwise, $G$ has the shape depicted on
the left of Figure~\ref{fig:seppathcompl}
\begin{figure}[t]
  \centering
  \includegraphics[width=.65\linewidth]{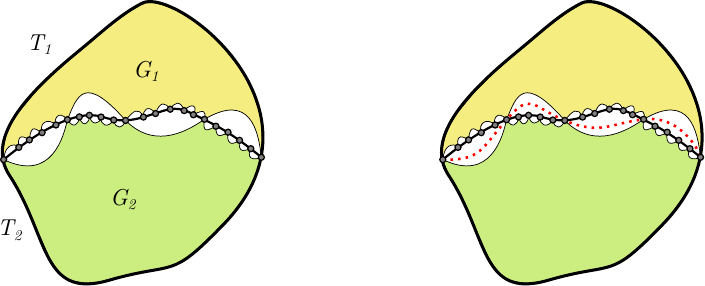}
  \caption{A frame-separating path and a completion for it.}
  \label{fig:seppathcompl}
\end{figure}
Let $T_1$ and $T_2$ be the two parts of $T$ which are separated by
$P$, The graphs $G_i$ with tuples $T_iP$ are then used as subinstances
for recursive calls.

Applying this process blindly is not sufficient. Indeed the recursive
calls may not be safe. More precisely, it is possible that $PT_i$ is
crossed in $G_i$ while $T$ is not in $G$ (for some $i\in\{1,2\}$). To
ensure safeness we first compute what we call a \emph{$P$-completion}.
Roughly, this is a set of edges on $P$'s vertices we add to $G$ such
that, using these edges, $P$ contains a shorter path $P'$ whose
associated subinstances are safe. In the easy cases, computing a
$P$-completion corresponds to adding the red dotted edges as on the
right of Figure~\ref{fig:seppathcompl}. Covering all cases is
delicate, but we eventually obtain a simple and efficient procedure
for computing $P$-completions.

\subsection{Contributions}
\label{ssec:intro:contrib}

Our three contributions are:
\begin{itemize}
\item an inductive characterisation of webs, via web composition (Section~\ref{sec:webs});
\item a direct proof of (a generalisation of) the two paths theorem (Section~\ref{sec:theorem}); 
\item a simple $O(nm)$ algorithm for the two disjoint paths problem (Section~\ref{sec:algorithm}).
\end{itemize}

\subsection{Outline}
\label{ssec:intro:outline}

We fix notations, make the above definitions precise, and provide a
few basic observations in Section~\ref{sec:defs}.

We study webs in Section~\ref{sec:webs}. We prove that their frames
are crossless (Proposition~\ref{prop:frame_crossless}) and that they
can be generated by repeated web compositions
(Proposition~\ref{prop:generating_webs_parallel}). We provide most
details for the sake of completeness, so that this section requires a
few pages.

We then move to the proof of the generalised two path theorem, in
Section~\ref{sec:theorem}.  We first prove maximality of webs amongst
crossless graphs, which is part of the easy direction in the
literature, and which is later instrumental for our algorithm.  We
finally prove that maximally crossless graphs are webs, which the
difficult implication in the works of
Thomassen~\cite{DBLP:journals/ejc/Thomassen80} and
Seymour~\cite{DBLP:journals/dm/Seymour80}. Despite its arrival late in
the text, this proof is self-contained: we only need to know from
Section~\ref{sec:webs} that webs are preserved by web composition
(Proposition~\ref{prop:web_stability}).

We finally present our algorithm for the two disjoint paths problem, in
Section~\ref{sec:algorithm}.

\clearpage
\section{Definitions}
\label{sec:defs}

\subsection{Standard notations and terminology}
\label{ssec:notation}

We only consider simple and undirected graphs, with the following
conventions. Given a graph $G$, we write $V(G)$ for its set of
\emph{vertices} and $E(G)$ for its set of \emph{edges}---two-elements
sets of vertices. We often use a graph to denote its underlying set of
vertices, thus writing $G$ for $V(G)$; similarly, we often write $x$
for the singleton set $\{x\}$. We note $n=|V(G)|$ the \emph{order} of
$G$ and $m=|E(G)|$ the \emph{size} of $G$. If $G'$ is another graph,
then $G\uplus G'$ is the \emph{disjoint union} of $G$ and $G'$.
Given $V\subseteq V(G)$ a subset of vertices, we note $G[V]$ for the
\emph{subgraph of $G$ induced by $V$}, and $G-V$ for the graph $G$
where all vertices of $V$ and incident edges have been removed. An
\emph{edge subgraph} of a graph $G$ is a graph $H$ with $V(H)=V(G)$
and $E(H)\subseteq E(G)$. The \emph{set of neighbours} of a vertex $x$
in a graph $G$ is denoted by $N_G(x)$.
If $F$ is a set of edges on $V(G)$ then $G+F$ is the graph obtained
from $G$ by adding the edges in $F$. Given two vertices $x$ and $y$
of $G$, we write $G/xy$ for the identification of $x$ and $y$
in $G$.
A \emph{walk} is a sequence of adjacent vertices; it is a \emph{path}
when no vertices are repeated; it is a \emph{cycle} if it contains at
least three vertices and exactly the first and last elements are
repeated. The \emph{length} of a path or cycle is its size, i.e., its
number of edges. An \emph{induced path} is a path that is induced as a
subgraph. Given a walk $P=x_1\dots x_k$, the \emph{interior} of $P$ is
$\mathring{P}=x_2\dots x_{k-1}$, and the \emph{reversal} of $P$ is
$\overline{P}=x_k\dots x_1$. Given a second walk $Q=y_1\dots y_r$
with, say, $x_i=y_j$, we note $Px_iQ$ for the walk from
$x_1\dots x_iy_{j+1}\dots y_r$. If $x_k=y_1$, we simply note $PQ$ for
$Px_kQ$. If $P$ and $Q$ are paths, $x_1=y_r$, and $x_s\neq y_t$ when
$1<s<k$ and $1<t<r$, then $PQ$ is a cycle.

Given two sets $H,K$ of vertices, a path $P=x_1\dots x_k$ is an
\emph{$H$-$K$ path} if $x_1\in H$, $x_k\in K$, and
$x_i\not\in H\cup K$ for all $1<i<k$. When $H=K$ we simply say that
$P$ is an \emph{$H$-path}. A third set of vertices $U$ is
\emph{$H$-connected to $K$} when there is a $H$-path from $U$ to $K$.
In this work, \emph{disjoint paths} always means vertex-disjoint
paths.
A subset $X$ of vertices \emph{separates} two sets of vertices $U$ and $V$
when all $U$-$V$ paths contain a vertex of $X$.

We write $K_n$ for the \emph{clique} of order $n$, i.e, the graph with
all edges between $n$ vertices.
A \emph{universal vertex} in a graph is a vertex adjacent to all other
vertices. For instance, the centre of a star graph is universal.  In
addition to the cliques, we define the following standard graph
families:
\begin{itemize}
\item the \emph{wheel graph} with $n\ge 4$ vertices is the graph $W_n$
  obtained from the cycle of size $n-1$ by the addition of a universal
  vertex;
\item the \emph{fan graph} with $n\ge 3$ vertices is the graph $F_n$
  obtained from the path of order $n-1$ by the addition of a universal
  vertex.
\end{itemize}
Examples of a wheel and a fan graph are shown in
Figure~\ref{fig:wheel_fan}.
\begin{figure}
  \centering
  \includegraphics{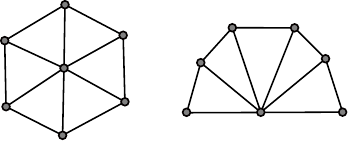}
  \caption{The wheel graph $W_7$ (left) and the fan graph $F_7$ (right).}
  \label{fig:wheel_fan}
\end{figure}
Observe that if $e$ is an edge in a wheel graph that is not incident
to the universal vertex, then $F_n\simeq W_n-e$.

\clearpage
\subsection{Generalised ribs and webs}
\label{ssec:ribs:webs}

\begin{definition}
  Given $k\ge 3$, a \emph{$k$-rib} is a planar graph such that in one
  of its plane embeddings, the outer face is a cycle of length $k$
  called its \emph{frame}, all inner faces are triangles, and all
  triangles are inner faces.

  A \emph{$k$-web} is a graph constructed from a $k$-rib $R$ by adding
  one clique $K_T$ per triangle $T$ of $R$ and making vertices of $K_T$
  adjacent to the vertices of $T$.
\end{definition}

In the sequel, when the size of the frame is not relevant, we often
simply call ribs and webs the graphs from this generalised definition.
Ribs are particular ``\emph{near-triangulations of the
  plane}''~\cite{TUTTE1980105} in which all triangles are faces, and
in the literature, webs are exactly the $4$-webs.
\begin{figure}
  \centering
  \includegraphics{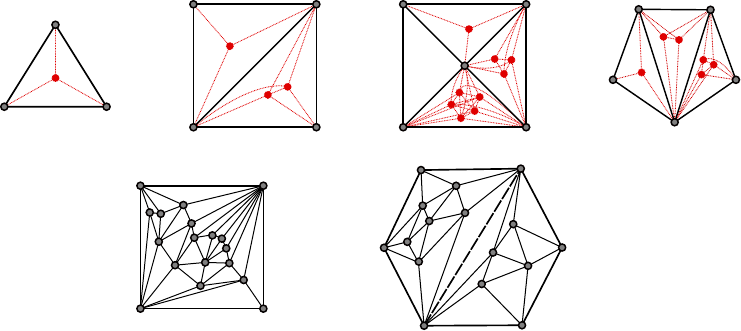}
  \caption{More examples of webs and ribs; dotted red vertices and
    edges are clique vertices and edges, while black ones are rib
    vertices and edges. The ribs of the webs on the top part are the
    only ribs of order three, four, and five.}
  \label{fig:web_examples}
\end{figure}
In addition to Figure~\ref{fig:smallribs} from the introduction,
examples of webs and ribs are depicted in
Figure~\ref{fig:web_examples}.
\begin{proposition}
  \label{prop:3_webs}
  The $3$-ribs are the triangles; the $3$-webs are the cliques of
  order at least $3$.
\end{proposition}

\begin{proposition}
  \label{prop:wheel_fan}
  Let $x$ be a vertex in a rib $R$ with frame $C$. The subgraph of $R$
  induced by $x$ and $N_R(x)$ is:
  \begin{itemize}
  \item a wheel graph when $x$ is not in $C$, or
  \item a fan graph when $x$ is in $C$.
  \end{itemize}
\end{proposition}
\begin{proof}
  Consider a plane embedding of $R$ in which $C$ bounds the outer
  face, all inner faces are triangles, and all triangles are inner
  faces. Let $H$ be the subgraph of $R$ induced by $x$ and its
  neighbours. In the embedding of $R$, the edges going out from $x$
  form a star whose centre is~$x$:

  \smallskip
  \begin{center}
    \includegraphics{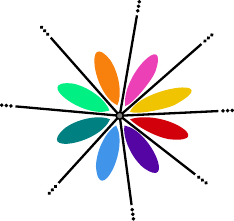}
  \end{center}

  \smallskip\noindent
  Each consecutive pair of edges bounds a common face represented by a coloured area
  in the drawing above. Either all such
  faces are inner, meaning are triangles, and $H$ is a wheel graph
  (and $x$ cannot bound the outer face: $x$ is not on $C$), or
  some of these faces are the outer face. The outer face being bounded
  by a cycle $C$, at most two edges incident to $x$ may bound the
  outer face. In this case $x$ is in $C$ and $H$ is a fan graph.
\end{proof}

\subsection{Crossings}
\label{ssec:crossings}

As explained in the introduction (Section~\ref{ssec:intro:thm}), the
appropriate generalisation of linkages is given by \emph{crossings}:
\begin{definition}
  Let $T=(x_1,\dots,x_k)$ be a tuple of at least three pairwise
  distinct vertices in a graph $G$. A \emph{crossing} of $T$ is a pair
  of disjoint $T$-paths, from $x_i$ to $x_j$ and $x_r$ to $x_s$,
  respectively, such that $i<r<j<s$. The tuple $T$ is \emph{crossed}
  if it admits a crossing; it is \emph{crossless} otherwise. The graph
  $G$ is \emph{maximally $T$-crossless} when $T$ is crossless and in
  any edge extension of $G$, $T$ crossed.
\end{definition}
Two examples of crossings are given in Figure~\ref{fig:crossings};
they are crossings in the bottom ribs of Figure~\ref{fig:web_examples}
but for a tuple different from their frames.
\begin{figure}
  \centering
  \includegraphics{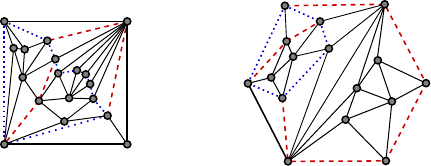}
  \caption{Two examples of crossings (red dashed paths) for some cycles (blue dotted cycles).}
  \label{fig:crossings}
\end{figure}

Observe that a graph is $2$-linked precisely when all tuples of four
distinct vertices are crossed; therefore, the standard two paths
Theorem~\ref{theo:main} is precisely our generalised
Theorem~\ref{theo:gen} when $|C|=4$. Moreover, in the case where
$|C|=3$, $C$ cannot be crossed, thus Theorem~\ref{theo:gen} holds
since $3$-webs are cliques by Proposition~\ref{prop:3_webs}.

\section{Web compositions and decompositions}
\label{sec:webs}

\subsection{Edge composition}
\label{par:edge_composition}

The graph operation consisting in taking the disjoint union of two
graphs and then identifying two edges $e$ and $e'$, one from each
graph, is called \emph{edge composition along $e$ and $e'$}.

Edge composition is illustrated in
Figure~\ref{fig:edge_compositions}.
\begin{figure}
  \centering
  \includegraphics{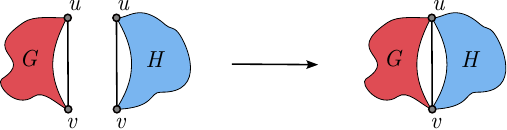}
  \caption{Edge composition of two graphs.}
  \label{fig:edge_compositions}
\end{figure}
Edge compositions of webs along edges of their frames always produce
webs: plane embeddings can be composed in this way, all triangles
remain faces and all inner faces remain triangles. If the identified
edges are denoted $e=xy$ and $e'=x'y'$ and the frames of the composed
webs are respectively $Pxy$ and $P'x'y'$, then the resulting web has
frame $P\overline{P'}$.

An example is the rib in the bottom right corner of
Figure~\ref{fig:web_examples}: the black dashed edge results from the
edge identification in the edge composition of two ribs. Similarly,
three of the six ribs from Figure~\ref{fig:smallribs} were obtained as
edge compositions of ribs.

In contrast, the rib in the bottom left corner of
Figure~\ref{fig:web_examples}, as well as the three other ribs from
Figure~\ref{fig:smallribs} cannot be expressed as an edge composition
of webs.

\subsection{Path composition}
\label{par:path_composition}

To decompose webs in general (except for the base case of 3-webs,
i.e., cliques), we generalise edge composition to identify subpaths of
frames instead of single edges.

If $P=x_1\dots x_l$ and $P'=x'_1\dots x'_l$ are non-empty paths of
length $l$ in two graphs $G$ and $H$, respectively, we define the
\emph{path composition} of $G$ and $H$ along $P$ and $P'$ as the
graph
\[
  (G\uplus H)/x_1x_1'/\dots/x_lx_l'.
\]
When $P=P'$ we simply talk of \emph{path composition along $P$}.

In the particular case where $P$ and $P'$ are reduced to one edge
each, that is $l=1$, the notion coincides with edge composition.

Unfortunately, path compositions of webs along subpaths of their
frames are not always webs (although this is true for edge
compositions, as explained above). A counter-example is given in
Figure~\ref{fig:path_composition_no_web}; to fix this issue, we have
to restrict to induced subpaths of the frames.
\begin{figure}
  \centering
  \includegraphics{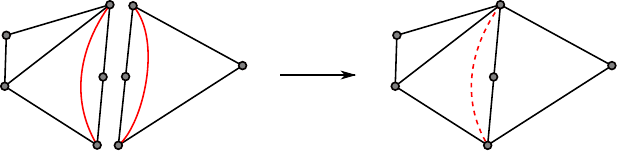}
  \caption{Path composition of two ribs. The resulting graph is not a
    rib (nor a web). Indeed, the dashed red edge, resulting from the
    collateral identification of the two bended edges on the left, is
    shared by three triangles; thus, at least one of these triangles
    cannot be a face in a plane embedding of the path composition.}
  \label{fig:path_composition_no_web}
\end{figure}

\subsection{Web composition}
\label{par:web_composition}

Let $G$ and $H$ be two webs and let $P$ and $P'$ be subpaths of their
respective frames of the same length. If $P$ and $P'$ are induced
paths in $G$ and $H$, respectively, we call \emph{web composition} of
$G$ and $H$ along $P$ and $P'$ the path composition of $G$ and $H$
along $P$ and $P'$.

Webs are stable under web composition:
\begin{proposition}
  \label{prop:web_stability}
  The web composition along $P$ and $P'$ of two webs $G$ and $H$ with
  respective frames $P_GP$ and $P_HP'$, is a web with frame $P_G\overline{P_H}$.
\end{proposition}
\begin{proof}
  Let $R_G$ and $R_H$ be the ribs of $G$ and $H$, respectively, and
  $R$ the web composition of $R_G$ and $R_H$ along $P$ and $P'$. We
  prove that $R$ is a rib.

  Any planar graph admits a plane embedding for which the outer face
  is sent exactly onto a given prescribed polygon of the plane (see
  \cite{article/lms/tutte60}).

  Consider plane embeddings of $R_G$ and $R_H$ for which the outer
  faces are sent onto two polygons whose intersection is precisely $P$
  and $P'$. As $P$ and $P'$ are induced, there are no parallel edges
  in the union of the plane embeddings of $R_G$ and $R_H$. Hence, this
  union is a plane embedding of $R$.

  As all inner faces of the embedding of $R_G$ or $R_H$ are triangles,
  so are inner faces of $R$. As $P$ and $P'$ are induced, a triangle
  of $R$ is either a triangle of $R_G$ or of $R_H$, but not of both.
  Hence every triangle of $R$ is a face in this embedding. So $R$ is a
  rib.

  Adding back the clique vertices of $G$ and $H$ on their respective
  triangles gives a web, proving that the web composition of $G$ and
  $H$ along $P$ and $P'$ is a web with rib $R$ and frame $P_G\overline{P_H}$.
\end{proof}

\subsection{Crossings in webs} An immediate yet important observation
is that in a web with rib $R$, the intersection with $R$ of a path
with endpoints in $R$ yields an induced path. Indeed, if such a path
$P$ uses a clique vertex $u$, then we can take the longest subpath $Q$
of $P$ containing $u$, and whose inner vertices are all clique
vertices. The endpoints of $Q$ are on a triangle $T$ of the rib. In
particular $Q$ can be replaced in $P$ by an edge of $T$:
\begin{center}
  \includegraphics{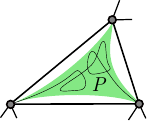}
\end{center}
In this figure the green triangular shape in the centre of the
triangle represents all the clique vertices of $T$.

Conversely, induced paths whose endpoints are on the rib are entirely
on the rib.

\medskip

Thomassen's argument~\cite{DBLP:journals/ejc/Thomassen80} proving that
frames of $4$-webs are not $2$-linked, generalises to $n$-webs:
\begin{proposition}
  \label{prop:frame_crossless}
  Frames of webs are crossless.
\end{proposition}
\begin{proof}
  Suppose by contradiction that the frame of a web is crossed. The
  union of the frame and the paths from a crossing yields a
  topological $K_4$.
  By intersecting paths with the rib, we can assume a topological
  $K_4$ in the rib.
  In a planar embedding of the rib in which the frame bounds the outer
  face, this provides an outerplanar plane embedding of $K_4$, a
  contradiction.
\end{proof}

\subsection{Edge decompositions}
\label{ssec:edge:decomp}

We characterise when a web with frame $C$ is expressible as the edge
composition of two webs using the $C$-connectivity of vertices; we
first need the following lemma.
\begin{lemma}
  \label{lem:4_frame_C_connectivity}
  Let $x$ be a vertex in a web $G$ with frame $C$ such that
  $|C|\ge 4$. There exists a $C$-path $P$ between non-consecutive
  vertices $c_x$ and $c_x'$ of $C$ such that $x$ is $C$-connected to
  all vertices of $P$.
\end{lemma}
\begin{proof}
  It is sufficient to prove the statement for ribs. Indeed, except
  perhaps for its endpoints, we have seen that induced paths in webs
  are paths in their ribs. Hence, two rib vertices are $C$-connected
  in $G$ if and only if they are $C$-connected in the rib of $G$; two
  clique vertices are $C$-connected in $G$ if and only if two of their
  rib neighbours are $C$-connected in the rib of $G$.

  When $x$ is a vertex in $C$, by Proposition~\ref{prop:wheel_fan},
  the neighbourhood of $x$ induces a fan graph $F$. Write $P$ for the
  path $F-x$ whose endpoints are the predecessor and successor of $x$
  on $C$. If $P$ is a $C$-path, then we take the endpoints of $P$ for
  $c_x$ and $c_x'$; $P$ is a $C$-path from $c_x$ to $c_x'$ whose
  vertices are all $C$-connected to $x$. If $P$ is not a $C$-path,
  then $P$ contains a vertex $y$ of $C$ which is adjacent to $x$.
  Since $y$ is inner in $P$, $x$ and $y$ are non-consecutive on $C$.
  We take $x$ and $y$ for $c_x$ and $c_x'$, and $Q=c_xc_x'$. The path
  $Q$ is a $C$-path from $c_x$ to $c_x'$ whose vertices are all
  $C$-connected to $x$.

  Assume that $x$ is not a vertex of $C$. Observe that ribs are
  connected, and consider an induced $C$-path $P$ from $x$ to a vertex $c$
  of $C$. By Proposition~\ref{prop:wheel_fan}, the neighbourhood of
  $c$ induces a fan graph in $G$. Hence, this neighbourhood provides
  paths from the penultimate vertex $u$ of $P$ to the two neighbours
  $d$ and $d'$ of $c$ on $C$, and these two paths $Q$ and $R$ are
  disjoint except possibly for $u$.
  \begin{figure}
    \centering
    \includegraphics{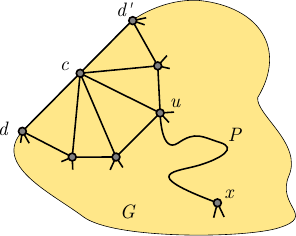}
    \caption{The case where $c_x=d$ and $c'_x=d'$ in the proof of Lemma~\ref{lem:4_frame_C_connectivity}.}
    \label{fig:rib_reaching_C}
  \end{figure}

  For $c_x$ and $c_x'$ we take the first vertices of $Q$ and $R$ which
  are on $C$, respectively (see Figure~\ref{fig:rib_reaching_C} for an
  illustration when $c_x$ and $c_x'$ are $d$ and $d'$).

  All vertices of $Q$ and $R$ are neighbours of $c$. For the sake of
  contradiction, suppose that $c_x$ and $c_x'$ are consecutive on
  $C$; $c_xc_x'cc_x$ is a triangle of $R$, and hence an inner face.
  The edge $cc_x$ bounds at least three faces in the plane embedding
  of $R$ in which the outer face is bounded by $C$, all inner faces
  are triangles and all triangles inner faces. In a plane embedding of
  a planar graph, an edge may bound at most two faces, a
  contradiction.

  The paths $PuQ$ and $PuR$ are $C$-paths from $x$ to $c_x$ and to
  $c_x'$, respectively. They induce a $C$-path from $c_x$ to $c_x'$
  whose vertices are all $C$-connected to $x$.
\end{proof}

Thanks to the previous lemma, we can prove the following
characterisation.  The second and third items in the proposition are
refinements of the first one, which are useful in the sequel.  We may
observe how this characterisation applies to the webs presented in
Figures~\ref{fig:smallribs} and~\ref{fig:web_examples}.
\begin{proposition}
  \label{prop:connectivity_webs}
  Let $G$ be a web with frame $C$. The following holds:
  \begin{enumerate}
  \item $G$ is the edge composition of two webs if and only if there
    is a pair of vertices of $G$ which are not $C$-connected.
  \item for all vertices $x,y$, either $x$ is $C$-connected to $y$, or
    there exists a pair of vertices $c,c'\in C-x-y$
    separating $x$ from $y$ and such that $G$ is the
    edge composition of two webs along $cc'$,
  \item for every vertex $x$ and every subpath $P$ of $C$, either $x$
    is $C$-connected to $P$, or there are vertices $c,c'$ of $C-P-x$
    separating $x$ from $P$ and such that $G$ is the edge composition
    of two webs along $cc'$.
  \end{enumerate}
\end{proposition}
\begin{proof}
  The first item of the statement is a direct consequence of the
  second one, which we prove now.

  Let $G$ be a web with frame $C$ and two vertices $x$ and $y$ which
  are not $C$-connected.

  By Proposition~\ref{prop:3_webs}, all pairs of vertices in $3$-ribs
  with frames triangles $T$ are $T$-connected. Hence we assume
  $|C|\ge 4$.

  By Lemma~\ref{lem:4_frame_C_connectivity}, $x$ is $C$-connected to
  at least two distinct non-consecutive vertices $c_x$ and $c_x'$ of
  $C$. Write $Q$ for a $C$-path from $c_x$ to $c_x'$ containing
  vertices all $C$-connected to $x$. Let $P$ and $P'$ be the two paths
  from $c_x$ to $c_x'$ such that $P\overline{P'}=C$.

  We prove that the vertices of $C$ which are $C$-connected to $y$ are
  either all on $P$ or all on $P'$. For the sake of contradiction,
  suppose that $y$ is $C$-connected to inner vertices $c_y$ and $c_y'$
  of $P$ and $P'$, respectively. Define a $C$-path $R$ from $P$ to
  $P'$ as follows:
  \begin{itemize}
  \item when $y$ is not on $C$, $R$ is a $c_y$-$c_y'$ path
    containing only vertices $C$-connected to $y$,
  \item when $y$ is on $C$, we assume without loss of generality that $y$
    is an inner vertex of $P$. For $R$ we take a $y$-$c_y'$ path.
  \end{itemize}
  Since $x$ and $y$ are not $C$-connected, $Q$ and $R$ are disjoint
  and $(Q,R)$ is a crossing of $C$, a contradiction.

  Hence, we assume that the vertices of $C$ to which $y$ is
  $C$-connected are all on $P$, and without loss of generality, we
  assume that $P$ is of minimum length such that its endpoints $c_x$
  and $c_x'$ are $C$-connected to $x$ and all vertices $C$-connected
  to $y$ are on $P$.

  We chose $c=c_x$ and $c'=c_x'$. We prove that $\{c,c'\}$ separates
  $\mathring{P}$ from $\mathring{P'}$. For the sake of contradiction,
  let $R$ be a $C$-path from $\mathring{P}$ to $\mathring{P'}$. Since
  $x$ is not $C$-connected to $\mathring{P}$, the vertices of $R$ are
  not $C$-connected to $x$ either. Recall that $Q$ is a $C$-path from
  $c$ to $c'$ whose vertices are all $C$-connected to $x$. The pair
  $(Q,R)$ is a crossing of $C$, a contradiction.

  Because $\{c,c'\}$ separates $\mathring{P}$ from $\mathring{P'}$, a
  plane embedding of the rib of $G$ in which the outer face is bounded
  by $C$, all inner faces are triangles and all triangles are inner
  faces, has the left-hand side shape below:
  \begin{center}
    \includegraphics{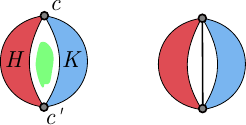}
  \end{center}
  Without being too formal on the manipulation of such embeddings, the
  green area must correspond to an inner face of the embedding,
  meaning a triangle. This face being incident to both $c$ and $c'$,
  that implies that $cc'$ is an edge of $G$. The fact that the rib of
  $G$ is given by an edge composition of two ribs along $cc'$ is then
  direct by looking at the inner faces and triangles of $H$ and $K$.
  Hence, $G$ is an edge composition of $H$ and $K$ in which the clique
  vertices are restored.

  We prove that $\{c,c'\}$ separates $x$ from $y$. For the sake of
  contradiction, let $R$ be an $x$-$y$ path in $G-c-c'$. Since $x$ and
  $y$ are not $C$ connected, $R$ contains at least one vertex of $C$.
  Write $c_1,\dots,c_n$ for the vertices of $R$ which are vertices of
  $C$: $R=Rc_1Rc_2\dots c_nR$. Since $x$ and $y$ are not $C$-connected
  to $\mathring{P}$ and $\mathring{P'}$, respectively, we have that
  $c_1$ and $c_n$ are on $\mathring{P'}$ and $\mathring{P}$,
  respectively. Hence, there exists an index $i$ such that $c_i$ is on
  $\mathring{P'}$ and $c_{i+1}$ on $\mathring{P}$. The path
  $c_iRc_{i+1}$ is a $C$-path from $\mathring{P'}$ to $\mathring{P}$
  in $G-c-c'$, a contradiction.

  The last item is proved along the same lines, replacing $y$
  with a subpath of $C$.
\end{proof}

\begin{corollary}
  \label{corol:path_web}
  In a web whose frame $C$ has length at least four, there is a
  $C$-path whose endpoints are non-consecutive vertices of $C$.
\end{corollary}
\begin{proof}
  As $|C|\ge 4$, consider two non-consecutive vertices $c,c'$ of $C$.
  If $c$ is $C$-connected to $c'$ then this provides a $C$-path whose
  endpoints are non-consecutive vertices of $C$. Otherwise, by
  Proposition~\ref{prop:connectivity_webs}, the web is given as an
  edge composition along some edge $uu'$ between nonconsecutive
  vertices of $C$. This edge is a $C$-path whose endpoints are
  non-consecutive vertices.
\end{proof}

\subsection{Web decompositions}
\label{par:web_decompositions}

While the bottom-left web in Figure~\ref{fig:web_examples} and the
third and sixth ribs in Figure~\ref{fig:smallribs} do not arise as
edge compositions, they can be expressed as web compositions. In fact,
we have the following decomposition result:
\begin{proposition}
  \label{prop:generating_webs_parallel}
  Every web whose frame has length at least four is a web composition
  of two webs.
\end{proposition}
\begin{proof}
  Let $G$ be a web with rib $R$ and frame $C$ with $|C|\geq 4$. By
  Corollary~\ref{corol:path_web} there is a $C$-path $P$ whose
  endpoints are non-consecutive vertices of $C$. We may assume without
  loss of generality that $P$ is an induced path, i.e., a path in $R$.
  Considering a plane embedding of the latter, the path $P$ separates
  $R$ into two planar graphs whose outer faces are bounded by the two
  cycles distinct from $C$ and induced by $C$ and $P$, all their faces
  are faces of $R$ and hence triangles, and all their triangles are
  triangles of $R$ and hence faces: those two planar graphs are ribs.
  Adding back the clique vertices of $G$ on their respective
  triangles, the resulting pair of webs generate $G$ by taking their
  web composition along $P$.
\end{proof}
Altogether, Propositions~\ref{prop:3_webs},~\ref{prop:web_stability},
and~\ref{prop:generating_webs_parallel} show that the webs are
precisely those graphs obtained by repeated web compositions, starting
from cliques of order at least three (whose frames are any triangle).
Note that such decompositions are not unique in general.

\section{Two paths theorem}
\label{sec:theorem}
\mbox{}\medskip

We now prove the generalised two paths theorem
(Theorem~\ref{theo:gen}).  We have already seen that web frames are
crossless (Proposition~\ref{prop:frame_crossless}); we first show that
the webs are maximally so. The proof of this result will be reused to
exhibit crossings in our presentation of the algorithm for the two
disjoint paths problem (Section~\ref{sec:computing_crossings_and_web_completions}).
\begin{lemma}
  \label{lem:web_are_edge_maximal}
  The frame of a web $G$ is crossed in all edge extensions of $G$.
\end{lemma}
\begin{proof}
  We reason by induction on the order of $G$. Let $C$ be the frame of
  $G$. Since $3$-webs are cliques (Proposition~\ref{prop:3_webs}) they
  cannot be extended with a new edge. We may thus assume $|C|\ge 4$,
  and hence by Proposition~\ref{prop:generating_webs_parallel}, that
  $G$ is the web composition of two webs $H$ and $K$ along some
  induced $C$-path $P$ of $G$. Let $C_H$ and $C_K$ be the frames of
  $H$ and $K$, respectively. By induction hypothesis, $C_H$ (resp.
  $C_K$) is crossed in any edge extension of $H$ (resp. $K$).

  Let $e=xy$ be an edge which is not in $G$. We prove that $C$
  is crossed in $G+e$.
  To build such a crossing we will rely on the following claim:

  \medskip
  \begin{minipage}{.9\linewidth}
    \begin{claim}
      \label{claim:connected}
      Every vertex of $H$ is $C_H$-connected to $C_H-P$ in $H$, and
      similarly in $K$.
    \end{claim}
    \emph{Proof.} As $P$ is induced in both $H$ and $K$, these
    webs cannot be expressed as some edge composition along some edge
    between non-consecutive vertices of $P$. The claim follows by the
    third item of Proposition~\ref{prop:connectivity_webs}\hfill$\lhd$
  \end{minipage}
  \medskip

  \noindent
  Assume that $x$ and $y$ are in $H-P$ and $K-P$, respectively. As $x$
  is $C_H$-connected to $C_H-P$ in $H$ and $y$ is $C_K$-connected to
  $C_K-P$ in $K$, we can assume to have an $x$-$(C_H-P)$ (resp.
  $y$-$(C_K-P)$) path $P_x$ in $H$ (resp. $P_y$ in $K$) disjoint from
  $P$. The pair $(P_xeP_y,P)$ is a crossing of $C$ in $G+e$.

  Otherwise $e$ has both its endpoints in either $H$ or $K$, say in $H$. By
  induction hypothesis, let $(P_1,P_2)$ be a crossing of $C_H$ in
  $H+e$. We reason by case analysis depending on which endpoints of $P_1$
  and $P_2$ are on $P$. Up-to symmetries between $P_1$ and $P_2$,
  there are five cases, represented in Figure~\ref{fig:five_cases}
  with the shape of the crossing of $C$ we build in $G+e$. The first
  three cases, in which both $P_1$ and $P_2$ have at least one end not
  on $P$ are self-explanatory and just require the use of $P$ to
  extend $P_1$ and/or $P_2$ to form a crossing of $C$ in $G+e$.
  \begin{figure}
    \centering
    \includegraphics{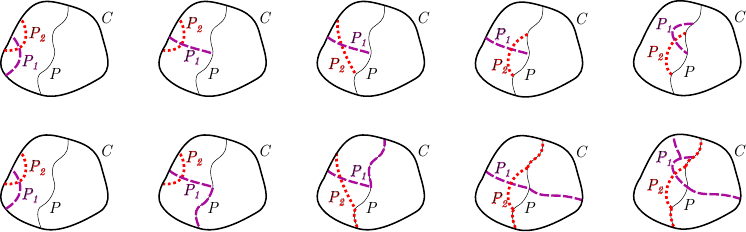}
    \caption{The five subcases in the proof of
      Lemma~\ref{lem:web_are_edge_maximal} (top), with crossings of
      the outer cycle we deduce in each case (bottom).}
    \label{fig:five_cases}
  \end{figure}

  In the fourth case of Figure~\ref{fig:five_cases}, if the end of
  $P_1$ which is on $P$ is denoted by $z$, then, since $z$ is
  $C_K$-connected to $C_K-P$ in $K$, there exists a $z$-$(C_K-P)$ path
  $P'$. Without loss of generality, assume that $P_2$ is from $x$ to
  $y$; the pair $(PxP_2yP,P_1zP')$ is a crossing of $C$ in $G+e$.

  In the fifth case, we first observe that $P_1$ and $P_2$ cannot be
  reduced to edges as $P$ is induced. All their inner vertices being
  $C_H$-connected to $C_H-P$, consider an associated path $P'$. Up-to
  replacing $P'$ by its subpath starting at the last vertex of
  either $P_1$ or $P_2$, we get, say, a $(C_H-P)$-$\mathring{P_1}$
  path inner disjoint from $P_1$, $P_2$, and $C_H$. Denote it by $P'$
  with $u$ its end on $P_1$. If $z$ is the end of $P_1$ that is
  surrounded by $P_2$'s endpoints on $P$, then just as above we get a
  $z$-$(C_K-P)$ path $P''$ in $K$ that is inner disjoint with $C_K$.
  Without loss of generality, assume that $P_2$ is a path from $x$ to
  $y$; $(P'uP_1zP'',PxP_2yP)$ is a crossing of $C$ in $G+e$.
\end{proof}

We finally prove the direct implication of the two paths theorem, which is
the difficult one in the literature.

\begin{lemma}
  \label{lem:maximallycrossless:web}
  Let $G$ be a graph with a cycle $C$. If $G$ is maximally
  $C$-crossless, then $G$ is a web with frame $C$.
\end{lemma}
\begin{proof}
  Let $G$ be a maximally $C$-crossless graph for a cycle $C$. Edges of
  $C$ cannot be used in crossings, so that $C$ must be a subgraph of
  $G$ by maximality. We proceed by induction on $|G|$ and we perform a
  case analysis on whether all pairs of distinct vertices of $C$ are
  $C$-connected or not.

  \medskip\noindent
  \emph{1. There are two vertices $x$ and $x'$ of $C$ which are not $C$-connected}
  \medskip

  Denote by $y$ and $z$ the two neighbours of $x$ on $C$. Trivially,
  $x,y,z$ are all $C$-connected to $x$. As $x$ is not $C$-connected to
  $x'$, we can consider $P'$ the longest portion of $C$ containing
  $x'$ whose inner vertices are all not $C$-connected to $x$. Let $P$
  be the path of $C$ such that $C=PP'$ and call $w,w'$ the shared
  endpoints of $P$ and $P'$. We are going to prove that $G$ is the
  edge composition of two webs along $ww'$; this situation is summarised in
  Figure~\ref{fig:first_3_1}.
  \begin{figure}
    \centering
    \includegraphics{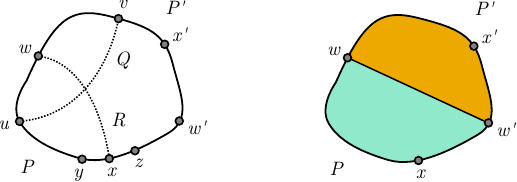}
    \caption{First case in the proof of Lemma~\ref{lem:maximallycrossless:web}; $x$ and $x'$ are not $C$-connected, $P$ and $P'$ both end at $w$ and $w'$.
      We eventually obtain the edge composition on the right.}
    \label{fig:first_3_1}
  \end{figure}

  To this end, we first prove that that $\{w,w'\}$ separates $C$. More
  precisely, that none of the inner vertices $u$ and $v$, respectively
  of $P$ and $P'$, are $C$-connected.

  If $u=x$ then inner vertices of $P'$ are not $C$-connected to $x$ by
  definition. Thus we assume $u\neq x$. For the sake of contradiction,
  let $Q$ be a $C$-path from $u$ to $v$. Let $w$ be the endpoint of
  $P$ such that $u$ is inner in $wPx$, and let $w'$ be the other end
  of $P$. By minimality of $P$, there exists an $x$-$w$ path $R$ inner
  disjoint from $C$. If $Q$ and $R$ were not disjoint, then they would
  induce an $x$-$v$ path inner disjoint from $C$, a
  contradiction. Otherwise, $(Q,R)$ is a crossing of $C$ in $G$, a
  contradiction.

  Thus $\{w,w'\}$ separates $C$. In particular adding the edge $ww'$
  cannot create a crossing of $C$ in $G$, so that by maximality, this
  edge must belong to $G$. Hence, $G$ is the edge composition of two
  graphs along $ww'$, respectively containing $P$ and $P'$. Quite
  directly, these two graphs are maximally crossless for $Pww'$ and
  $P'ww'$, respectively. By induction hypothesis they are webs and, by
  Proposition~\ref{prop:web_stability}, $G$ is a web.

  \medskip
  \noindent
  \emph{2. All pairs of distinct vertices of $C$ are $C$-connected}
  \medskip

  Take consecutive vertices $x,y,z$ on $C$ with an $x$-$z$ path $P$
  inner disjoint from $C$. For such a path $P$ we write
  $\mathcal{C}_C^P$ for the component of $C-\{x,y,z\}$ in $G-P$. We assume
  that $P$ is minimal with regard to the measure
  \[
    \mu(P)=(|G|-|\mathcal{C}_C^P|,|P|)
  \]
  where the order on $\mathbb{N}\times\mathbb{N}$ is the lexicographic
  product of the usual order $<$. Said otherwise we choose $P$ to
  maximise the order of the component of $G-P$ containing $C-\{x,y,z\}$
  while minimising the length of $P$. In particular, $P$ is induced.

  Take $u$ an inner vertex of $P$. We prove that $u$ is
  $C\cup P$-connected to $C-\{x,y,z\}$. For the sake of contradiction,
  assume not. We build a crossing of $C$ in $G$.

  Let $v,w\in P$ be such that $vPw$ is the longest subpath of $P$
  containing $u$ as an inner vertex and whose inner vertices are all
  not $C\cup P$-connected to $C-\{x,y,z\}$. As $P$ is induced and as
  $vPw$ contains at least $u$ as an inner vertex, $G$ does not
  contain~$vw$.

  Consider $G+vw$. By maximality of $G$, $C$ is crossed in $G+vw$. Let
  $(P_1,P_2)$ be a crossing of $C$ in $G+vw$, with, say,
  $vw\in E(P_1)$. The path $P_2$ must contain some inner vertex of
  $vPw$ as otherwise $(P_2,P_1vPwP_1)$ would be a crossing of $C$ in
  $G$.

  One end of $P_2$ at least must be on $C-y$. Without loss of
  generality, assume this is the case of its last vertex. Let $s$ be
  the last vertex in $vPw$ seen by $P_2$. As the last vertex of $P_2$
  is on $C-y$, and since $s$ is not $C\cup P$-connected to
  $C-\{x,y,z\}$, $sP_2$ must either intersect $P$ at one of its inner
  vertex $t$, or its last vertex is $x$ or $z$. In any case, there is
  a vertex $t$ of $sP_2$ appearing on $P$, after $s$. This situation
  is depicted on the left of Figure~\ref{fig:maincase}.
  \begin{figure}[t]
    \centering
    \includegraphics{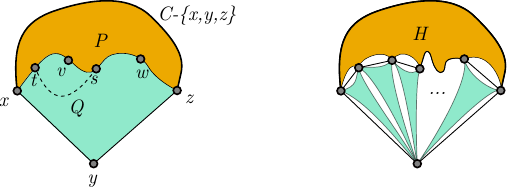}
    \caption{Second in the proof of Lemma~\ref{lem:maximallycrossless:web}, and shape eventually obtained.}
    \label{fig:maincase}
  \end{figure}
  
  By maximality of $s$ it must be outside of
  $vPw$, providing an $s$-$t$ path $Q$ such that either $v\in sPt$ or
  $w\in sPt$. But then, $P$ can be replaced by $PsQtP$ giving a path
  with lower measure $\mu$ as $\mathcal{C}_C^{PsQtP}$ contains $v$ or
  $w$ on top of all vertices of $\mathcal{C}_C^P$. This is a
  contradiction.

  The fact that all inner vertices of $P$ are $C\cup P$-connected to
  $C-\{x,y,z\}$ has the following implication: components of $G-P-y$
  that are incident to $y$ are incident to at most two vertices of $P$
  and those are consecutive in $P$. Indeed if such a component could
  join two non-consecutive vertices $s$ and $t$ of $P$ in a
  $C\cup P$-path $Q$ not intersecting $\mathcal{C}_C^P$, then $PsQtP$
  would give a better path than $P$ with regard to the measure $\mu$.
  Said otherwise, we have proved that $G$ has the shape depicted on
  the right of Figure~\ref{fig:maincase}.
  
  In other words, $G$ is the path composition along $P$ of an edge
  subgraph of a web (made of the green triangular shapes and whose rib
  is a fan graph) with frame $zyxP$ with a graph $H$ (the orange area
  above the triangular shapes) containing $PzCx$ as a cycle. It is
  easily proved that $PzCx$ is crossed in any edge extension of $H$,
  hence $H$ is a web with frame $PzCx$ by induction. By
  Proposition~\ref{prop:web_stability}, since $G$ is the path
  composition along an induced path of $H$ and an edge subgraph of a
  web, $G$ is an edge subgraph of a web with frame $C$. By maximality
  of $G$ and Lemma~\ref{lem:web_are_edge_maximal}, $G$ is a web with
  frame $C$.
\end{proof}

Combining Lemmas~\ref{lem:web_are_edge_maximal} and~\ref{lem:maximallycrossless:web} with
Proposition~\ref{prop:frame_crossless}, 
Theorem~\ref{theo:gen} is proved.

\section{Two disjoint paths problem}
\label{sec:algorithm}
\mbox{}\medskip

We now turn to the algorithmic problem of finding two disjoint paths.
Let $G$ be a graph with a tuple $T$ of pairwise distinct vertices. A
\emph{web completion} is a set $F$ of edges such that $G+F$ is a web
with frame $T$.
By Theorem~\ref{theo:gen}, web completions exist if and only if $T$ is
crossless; we solve the following:

\medskip%
\noindent\textsc{Two Disjoint Paths or Web Completion}\\
\noindent\textbf{Input:} A graph with a tuple of distinct vertices.\\
\noindent\textbf{Output:} A crossing or a web completion.%
\medskip%

In the next two sections we only consider the decision version of this
problem. We discuss the actual computation of crossings and web
completions in
Section~\ref{sec:computing_crossings_and_web_completions}, and
algorithmic complexity in Section~\ref{sec:complexity}.

\subsection{Method}

Let us abstract the two main cases in the proof of
Lemma~\ref{lem:maximallycrossless:web}: given a graph which is
maximally crossless for a cycle $C$, we find an induced $C$-path $P$
which separates $C$ into two paths $P_1$ and $P_2$ corresponding to
two components $V_1$ and $V_2$ of $G-P$; up to some assumptions on
$P$, the graphs $G[V_1\cup P]$ and $G[V_2\cup P]$ are maximally
crossless and we conclude by induction.

We turn this inductive argument into a recursive algorithm for finding
crossings. The minimality assumption on the separating path in the
proof of Lemma~\ref{lem:maximallycrossless:web} (in particular with
respect to the measure $\mu$) is not suitable to
get an efficient algorithm; the difficulty is thus to find a
separating path with good enough properties.

Henceforth we assume a graph $G$ with a tuple $T=(x_1,\dots,x_k)$ of
distinct vertices. A set of edges $F$ is called \emph{safe} for $T$
if, when $T$ is crossed in $G+F$, then $T$ is crossed in $G$.
If $T$ is crossed then all sets of edges are safe; moreover the edges
$x_ix_{i+1}$ ($1\le i\le k-1$) and $x_kx_1$ are all safe for $T$ since
they cannot be used in a crossing. Hence we assume that $T$ induces a
cycle $C$ in $G$.

\clearpage
We use the following divide-and-conquer strategy:
\begin{itemize}
\item if all $C$-paths of $G$ are between consecutive vertices of $C$
  then $C$ is crossless;
\item if there is a $C$-path $P$ between non-consecutive vertices of
  $C$ such that $C-P$ is connected in $G-P$, then $C$ is crossed;
\item otherwise, use a separating $C$-path between non-consecutive
  vertices of $C$ to find two graphs $G_1$ and $G_2$ with cycles
  $C_1$ and $C_2$ such that:
  \begin{itemize}
  \item $G$ is the path composition of $G_1$ and $G_2$ along $P$, and
  \item $C$ is crossed in $G$ if and only if $C_1$ is crossed in $G_1$
    or $C_2$ is crossed in~$G_2$.
  \end{itemize}
  Solve the corresponding subinstances recursively, and answer accordingly.\medskip
\end{itemize}
To obtain the two subinstances, we first prove the following proposition.
\begin{proposition}
  \label{prop:subinstances:prelim}
  Let $P$ be a $C$-path of $G$ between non-consecutive vertices of
  $C$, and let $P_1$ and $P_2$ be the two distinct subpaths of $C$ sharing
  their endpoints with $P$ and such that $C=P_1\overline{P_2}$.

  If $P$ separates $\mathring{P_1}$ from $\mathring{P_2}$ then $P_1$
  and $P_2$ are contained in subgraphs $G_1$ and $G_2$ of $G$,
  respectively, and $G$ is the path composition of $G_1$ and $G_2$
  along $P$.
\end{proposition}
\begin{proof}
  Observe that $\mathring{P_1}$ and $\mathring{P_2}$ are both
  non-empty as the endpoints of $P$ are non-consecutive on $C$.
  The statement follows by studying components of $G-P$. Any such
  component containing vertices of $\mathring{P_i}$ is associated to
  $G_i$. All other components as well as edges between non-consecutive
  vertices of $P$ are arbitrarily distributed between $G_1$ and $G_2$.
  The graph $G_i$ is defined as the subgraph of $G$ induced by the
  union of $P$ with the components of $G-P$ associated with $G_i$ to
  which any edge between non-consecutive vertices of $P$ it is not
  associated with is removed. The fact that $G$ is indeed the path
  composition of $G_1$ and $G_2$ along $P$ follows easily from $P$
  being a separator of $\mathring{P_1}$ and $\mathring{P_2}$.
\end{proof}

Following the divide-and-conquer strategy described above, the pairs
$(G_i,C_i)$ (with $C_i=P_iP$) are regarded as subinstances.
Unfortunately, in general it is false that $C_1$ or $C_2$ is crossed
(resp. in $G_1$ or $G_2$) if and only if $C$ is crossed in $G$.

The left-hand side example of Figure~\ref{fig:wrong_sub_instances}
shows a graph $G$ given as the path composition of two graphs $G_1$
and $G_2$, each corresponding to a coloured area of the drawing (two
edges do cross so this is indeed a drawing and not a plane embedding).
The path $P$ is the one in the intersection of the two graphs. The
cycle $C$ is the outer $4$-cycle. The cycle $C_1$ is crossed in $G_1$
while $C$ is crossless in $G$. On the right-hand side of
Figure~\ref{fig:wrong_sub_instances} we have displayed the only web
completion of $G$ (see how it can be built by adding one clique of
order two on some triangle of the only rib of order five with a frame
of length four in Figure~\ref{fig:web_examples}). It cannot be
expressed as the path composition of two graphs along $P$: the
subinstances induced by $P$ in $G$ do not behave well with regard to
web completions.
\begin{figure}
  \centering
  \includegraphics[width=.6\linewidth]{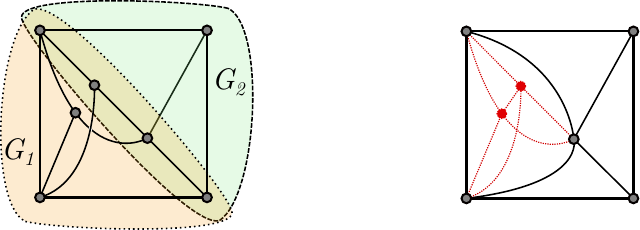}
  \caption{A graph and its only web completion (with respect to the
    outer $4$-cycle).}
  \label{fig:wrong_sub_instances}
\end{figure}

Our solution consists in ensuring a connectivity property on $P$.
Given $i\in\{1,2\}$, we say that a vertex $y\in P$ \emph{is on
  $P$-side $i$} when $y$ and $\mathring{P_i}$ are $P$-connected.
We refine the previous proposition as follows:
\begin{proposition}
  \label{prop:subinstances}
  Let $P$ be a $C$-path of $G$ between non-consecutive vertices of
  $C$, and let $P_1$ and $P_2$ be the two distinct subpaths of $C$ sharing
  their endpoints with $P$ such that $C=P_1\overline{P_2}$.

  If $P$ separates $\mathring P_1$ from $\mathring P_2$ and all vertices of $P$ are on
  both $P$-sides, then $C$ is crossed in $G$ if and only if $P_1P$ is
  crossed in $G_1$ or $P_2P$ is crossed in $G_2$, where $G_1$ and $G_2$
  are the subgraphs obtained by
  Proposition~\ref{prop:subinstances:prelim}.
\end{proposition}
\begin{proof}
  Suppose that $\mathring P_1$ and $\mathring P_2$ are separated by $P$, and that all
  vertices of $P$ lie on both $P$-sides. We prove the contrapositives
  of both implications.
  \begin{itemize}
  \item Assume that $P_iP$ is crossless in $G_i$ for every
    $i\in\{1,2\}$. By Theorem~\ref{theo:gen}, $G_i$ has a web
    completion $F_i$.

    In $G_i+F_i$, the path $P$ must be induced. Indeed, if $xy$ were
    an edge of $G_i+F_i$ between non-consecutive vertices of $P$, then
    we could consider $z$ an inner vertex of $xPy$. As $z$ is on both
    $P$-sides, there is in $G_i$ a $\mathring{P_i}$-$z$ path $Q$. The
    pair $(Q,xy)$ is a crossing of $P_iP$ in the web $G_i+F_i$, a
    contradiction.

    By Proposition~\ref{prop:web_stability} the web
    composition of $G_1+F_1$ and $G_2+F_2$ along $P$ is a web,
    proving that $G$ has a web completion in the form of
    $F_1\cup F_2$, and as a consequence of Theorem~\ref{theo:gen}
    that $C$ is crossless.
  \item Suppose $C$ is crossless in $G$, and consider a component $D$
    of $G-P$ containing vertices of neither $P_1$ nor $P_2$. It is
    incident to at most two vertices of $P$ and these must be
    consecutive. Indeed, if that was not the case, then $D$ would
    induce an $x$-$y$ path $Q$ between non-consecutive vertices of
    $P$. Let $z$ be an inner vertex of $xPy$. As $z$ is on both
    $P$-sides, there are $\mathring{P_1}$ -$z$ ad $z$-$\mathring{P_2}$
    paths $R$ and $S$ inner disjoint with $C$, $P$, and $Q$. The pair
    $(PxQyP,RS)$ is a crossing of $C$ in $G$, a contradiction.

    The same method works when $D$ is replaced with an edge between
    non-consecutive vertices of $P$.
    In particular, $P$ is induced in $G$.

    Let $\mathcal{D}$ be the set of all components of $G-P$ containing
    vertices of neither $P_1$ nor $P_2$. $C$ is crossless in
    $G-\mathcal{D}$, so by Theorem~\ref{theo:gen} we let $F$ be an
    associated web completion. As edges between non-consecutive
    vertices of $P$ make $C$ be crossed, $P$ remains induced in
    $(G-\mathcal{D})+F$. Hence $P$'s vertices are all in the rib of
    $(G-\mathcal{D})+F$. This gives webs $H_i$ with frames $P_iP$ such
    that $(G-\mathcal{D})+F$ is their web composition along $P$.
    Obviously $H_i$ is a web containing $G_i-\mathcal{D}$. For each
    component $D$ of $\mathcal{D}$ incident to, say, an edge $xy$ of
    $P$, if its vertices are in $G_i$ then we make them all into a
    clique, and adjacent to all vertices of the triangle of $H_i$'s
    rib containing $xy$ and of the clique of this triangle.

    This construction provides a web completion of both $G_1$ and
    $G_2$ proving that $P_1P$ and $P_2P$ are crossless, respectively in
    $G_1$ and $G_2$.
  \end{itemize}
\end{proof}

The example of Figure~\ref{fig:wrong_sub_instances} however provides a
graph where no $C$-path $P$ may be used to apply this proposition and
get two subinstances. We prove that we can always add edges so
that we can eventually apply this proposition.

Given a $C$-path $P$ of $G$ we call \emph{$P$-completion} a safe set
$F$ of edges such that $P+F$ contains a $C$-path $P'$ whose vertices
are all on both $P'$-sides in $G+F$.

\begin{proposition}
  \label{prop:P_completion_naive}
  Let $G$ be a graph with a cycle $C$. For every $C$-path $P$ between
  non-consecutive vertices of $C$, there exists a $P$-completion.
\end{proposition}
\begin{proof}
  Assume $P=y_1\dots y_l$. If $C$ is crossed in $G$ then $y_1y_l$ is a
  $P$-completion. If instead $C$ is crossless, then by
  Theorem~\ref{theo:gen} there is a set $F$ of edges such that $G+F$
  is a web. Any $C$-path $P'$ of $G+F$ which is induced and contained
  in $P+F$ can then be used to prove that $F$ is a $P$-completion (see
  Claim~\ref{claim:connected} for details on proving that vertices of
  $P'$ lie on both $P'$-sides).
\end{proof}

Equipped with these tools, we introduce our general method in
Algorithm~\ref{algo:general_method}. We assume that a call to
$\mathbf{P\_completion}(G,C,P)$ computes a $P$-completion $F$ with
associated $C$-path $P'$, and returns the pair $(G+F,P')$. By
Proposition~\ref{prop:subinstances} and induction on the size of the
recursive instances, the method is both correct and terminating.

It remains to implement $\mathbf{P\_completion}$. A naive computation
of $P$-completions, along the lines of the proof of
Proposition~\ref{prop:P_completion_naive}, amounts to solving the two
disjoint paths problem. Our method is only of interest if
$P$-completions can be computed efficiently, which is the topic of the
next section.
\begin{algorithm}[t]
  \caption{General method for the two disjoint paths decision problem.}
  \label{algo:general_method}
  \begin{algorithmic}[1]
    \Procedure{\bf{2DP}}{$G,C$}\Comment{decides if $C$ is crossed in $G$ or not}
    \If{there are no $C$-paths between non-consecutive vertices of $C$}
    \State \textbf{return} False
    \Else
    \State $P\gets$ a $C$-path between non-consecutive vertices of $C$
    \If{$P$ does not separate $C$ in $G$}
    \State \textbf{return} True
    \Else
    \State $G,P\gets\mathbf{P\_completion}(G,C,P)$
    \State $G_i,C_i\gets$ two subinstances induced by $G$, $P$ and $C$
    \State \textbf{return} \textbf{(2DP($G_1$,$C_1$)} $\vee$ \textbf{2DP($G_2$,$C_2$))}
    \EndIf
    \EndIf
    \EndProcedure
  \end{algorithmic}
\end{algorithm}

\subsection{Computing \emph{P}-completions efficiently} Recall that we are
working with a graph $G$ with a cycle $C$ and a $C$-path $P$
between non-consecutive vertices of $C$. The goal is to compute a
$P$-completion efficiently. As before we let $P_1,P_2$ be the two
distinct subpaths of $C$ with the same endpoints as $P$ and such that
$C=P_1\overline{P_2}$.

The core idea behind computing a $P$-completion is to analyse which
$P$-sides each vertex of $P$ lies on. This reveals some gaps as in the
left-hand side graph shape of Figure~\ref{fig:intuition_P_completion},
where the outer cycle is $C$, and the central path with explicit
vertices is $P$. At first sight, the red dotted edges form a $P$-completion.
\begin{figure}[t]
  \centering
  \includegraphics{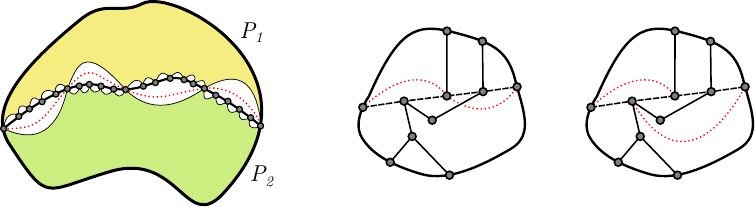}
  \caption{Intuition on how $P$-completions are computed.}
  \label{fig:intuition_P_completion}
\end{figure}
The issue is that the red dotted edges may not be well-defined,
because there might be components of $G-P$ intersecting neither
$\mathring{P_1}$ nor $\mathring{P_2}$, and also because the gaps may
not intersect on vertices which lie on both $P$-sides. The two other
graphs of Figure~\ref{fig:intuition_P_completion} depict two possible
choices of red dotted edges in a concrete graph (where $P$ is the
dashed straight path). In the first one, the right-most red dotted
edge is not safe: it makes the outer cycle crossed. In the second one,
the two red dotted edges cannot both be used in a $C$-path. Our
solution, described in Lemma~\ref{lem:computing_P_completions} below,
is to add red edges one after the other, recomputing sides at each
step. On the previous example, we will obtain that adding one red edge
from the last graph (but not both) is enough to get a $P$-completion.

Write $y_1\dots y_l$ for $P$. Since $P$ connects non-consecutive
vertices of $C$, both $\mathring{P_1}$ and $\mathring{P_2}$ are
nonempty. Moreover, as $y_1$ (resp. $y_l$) is adjacent to a vertex of
both $\mathring{P_1}$ and $\mathring{P_2}$, $y_1$ (resp. $y_l$) is on
both $P$-sides. Hence the following are well-defined: 1. the largest
prefix $Py$ of $P$ whose vertices are all on both $P$-sides; 2. when
$y\neq y_l$, the first vertex $z$ on $P$ found after $y$ which is on a
given $P$-side.

The cornerstone for computing $P$-completions efficiently is the
following lemma. Its proof is involved although the resulting
algorithm is quite simple in the end.
\begin{lemma}
  \label{lem:computing_P_completions}
  Assume that not all vertices of $P$ lie on both $P$-sides.

  Let $j<l$ be the greatest index such that all vertices of $Py_j$ are
  on both $P$-sides.
  
  Let $r>j$ be the lowest index such that $y_r$ is on at least one
  $P$-side.

  Let $s\ge r$ be either:
  \begin{itemize}
  \item $s=r$ if $y_r$ is on both $P$-sides, or
  \item the lowest index such that $y_s$ lies on the opposite $P$-side of $y_r$.
  \end{itemize}
  Then the edge $y_jy_s$ is safe, and in the graph $G+y_jy_s$, all
  vertices of $Py_jy_s$ lie on both $Py_jy_sP$-sides.
\end{lemma}
\begin{proof}
  To prove that all vertices of $Py_jy_s$ lie on both $Py_jy_sP$-sides
  in $G+y_jy_s$ we only need proving that $y_s$ is on both
  $Py_jy_sP$-sides. The remainder follows from the assumption that all
  vertices of $Py_j$ lie on both $P$-sides in $G$.

  Without loss of generality, we assume that $y_r$ is on $P$-side $1$
  in $G$. The vertex $y_s$ is on $P$-side $2$ (we may have $r=s$).
  These two facts provide a $\mathring{P_1}$-$y_r$ path $Q$ inner
  disjoint with $P\cup C$ and a $\mathring{P_2}$-$y_s$ path $R$ inner
  disjoint with $P\cup C$. The paths $R$ and $Qy_rPy_s$ prove that
  $y_s$ is on both $Py_jy_sP$-sides in $G+y_jy_s$.

  \medskip

  We now prove that $y_jy_s$ is safe.
  If $C$ is crossed in $G$ then it is so in $G+y_jy_s$.
  Otherwise, $C$ is crossless in $G$ and we prove that it is crossless
  in $G+y_jy_s$ too.
  By Theorem~\ref{theo:gen}, $G$ has a web completion $F$. Below,
  we modify $F$ to get a $P$-completion that contains $y_jy_s$,
  proving this edge is safe for $C$. In $G+F$, we let $P'$ be an induced
  $C$-path extracted from $P+F$. We first prove four claims
  relating $P'$ with $P$:

  \medskip
  \begin{minipage}{.9\linewidth}
    \begin{claim}
      \label{claim:yjyr_is_safe}
      The edge $y_jy_r$ is safe in $G$.
    \end{claim}
    \emph{Proof.} For the sake of contradiction, let $(Q,R)$ be
    a crossing of $C$ in $G+y_jy_r$.

    Since $C$ is crossless in $G$, $(Q,R)$ is not a crossing of $C$ in
    $G$ and (without loss of generality) $Q$ contains the edge
    $y_jy_r$. As vertices $y_t$ with $j<t<r$ have no $P$-sides,
    $\{y_j,y_r\}$ separates them from $C$. Hence, since $R$ has its
    endpoints on $C$, $R$ does not intersect $y_jPy_r$. The pair
    $(Qy_jPy_rQ,R)$ is a crossing of $C$ in $G$, a
    contradiction.\hfill$\lhd$
  \end{minipage}

  \medskip  
  \begin{minipage}{.9\linewidth}
    \begin{claim}
      \label{claim:side_preservation}
      If some vertex $z$ has $P$-side $i$ in $G$ then it has $P'$-side
      $i$ in $G+F$.
    \end{claim}
    \emph{Proof.} Indeed, as $V(P')\subseteq V(P)$, a
    $C\cup P$-path from $z$ to $\mathring{P_i}$ in $G$ is a
    $C\cup P'$-path from $z$ to $\mathring{P_i}$ in $G+F$.\hfill$\lhd$
  \end{minipage}
  
  \medskip
  \begin{minipage}{.9\linewidth}
    \begin{claim}
      \label{claim:prefix}
      All vertices $y_t$ with $1\le t\le j$ are on $P'$. In particular,
      the prefix of $P'$ of length $j-1$ is $Py_j$.
    \end{claim}
    \emph{Proof.} For the sake of contradiction, let $y_t$ with
    $1\le t\le j$ be a vertex which is not on $P'$. Since $y_t$ is on
    both $P$-sides, let $Q_1$ and $Q_2$ be $P\cup C$-paths in $G$,
    respectively from $\mathring{P_1}$ to $y_t$ and from $y_t$ to
    $\mathring{P_2}$. The pair $(Q_1Q_2,P')$ is a crossing of $C$ in
    $G+F$, a contradiction.

    We prove that the prefix of $P'$ of length $j-1$ is $Py_j$.

    For the sake of contradiction, assume that the prefix of $P'$ of
    length $j-1$ is not $Py_j$ and let $y_a$ be the first vertex on $P'$
    such that $P'=Py_ay_{a'}P'$ with $a'\neq a+1$. As $y_{a+1}$ is on
    both $P$-sides, we can consider $P\cup C$-paths $Q$ and $R$
    respectively from $\mathring{P_1}$ to $y_{a+1}$ and from $y_{a+1}$
    to $\mathring{P_2}$. The pair $(Q_1Q_2,Py_ay_{a'}P)$ is a crossing
    of $C$ in $G+F$, a contradiction.\hfill$\lhd$
  \end{minipage}
  
  \medskip
  \begin{minipage}{.9\linewidth}
    \begin{claim}
      \label{claim:j_plus_1}
      The $j+1$th vertex of $P'$ is a vertex $y_t$ with $j<t\le s$.
    \end{claim}
    \emph{Proof.} For the sake of contradiction, assume not. By
    Claim~\ref{claim:prefix}, the $j+1$th vertex of $P'$ is $y_t$ for
    $t>s$. Since $y_r$ and $y_s$ are respectively on, say, $P$-sides $1$
    and $2$, we can consider $P\cup C$-paths $Q$ and $R$ respectively
    from $\mathring{P_1}$ to $y_r$ and from $\mathring{P_2}$ to $y_s$.
    The pair $(Qy_rPy_sR,Py_jy_tP)$ is then a crossing of $C$ in $G+F$,
    a contradiction.\hfill$\lhd$
  \end{minipage}
  \bigskip

  We are now back to proving that $y_jy_s$ is safe. By
  Claim~\ref{claim:yjyr_is_safe} we can assume $r\neq s$, implying, by
  definition of $r$ and $s$, that $r$ is on exactly one $P$-side in
  $G$. Without loss of generality, assume $y_r$ is on $P$-side $1$ but
  not $2$ in $G$. This implies that $y_s$ is on $P$-side $2$ in $G$
  and is the only such vertex among $y_{j+1},\dots,y_s$.
  Claims~\ref{claim:prefix} and~\ref{claim:j_plus_1} prove that
  $P'=Py_jy_ty_{t'}P'$ with $j<t\le s$.

  We prove the safeness of $y_jy_s$ by induction on the length of
  $P'$. If $s=t$ then obviously $F$ contains $y_jy_s$ and this edge is
  safe. Assume $s-t>0$. It implies that $y_t$ is not on $P$-side $2$
  in $G$.

  The idea of the proof is to use the fact that none of the vertices
  in $y_{j+1},\dots,y_t$ lie on $P$-side $2$ in $G$ to ``cut'' some
  edges of $F$ between these vertices and $P$-side $2$, and replace
  them by an edge closer to $y_jy_s$:
  \smallskip
  \begin{center}
    \includegraphics[width=.65\linewidth]{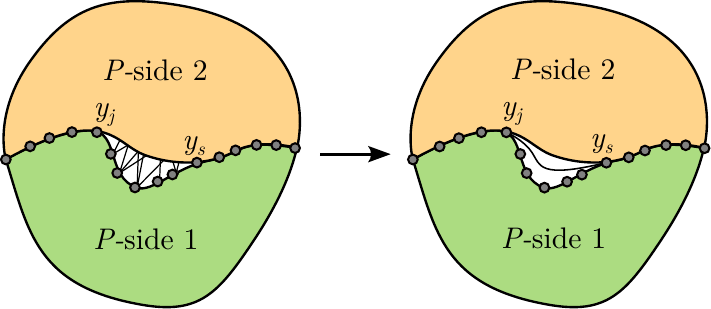}
  \end{center}
  
  \smallskip\noindent Assuming this illustration takes the planar
  embedding of the rib of $G$ into account, adding $y_jy_s$ is safe
  after the removal of the edges drawn in the middle of the graph on
  the left hand-side.

  This illustration works when $P'=P$ and the vertices in
  $y_{j+1},\dots,y_{s-1}$ are directly connected to vertices on
  $P$-side $2$ in $G$ (and hence the edges we remove are all directly
  incident to vertices in $y_{j+1},\dots,y_{s-1}$). Neither one of
  these assumptions hold in general. Instead we prove that $y_jz$ is
  safe, for $z$ a vertex appearing after $y_t$ on $P'$, allowing us to
  apply the induction hypothesis on $P'y_jzP'$.

  By Proposition~\ref{prop:wheel_fan}, in $G+F$, $y_t$ and neighbours
  of $y_t$ that are on $P'$-side $2$ and on the rib of $G+F$ describe
  a fan graph:

  \smallskip
  \begin{center}
    \includegraphics[width=.3\linewidth]{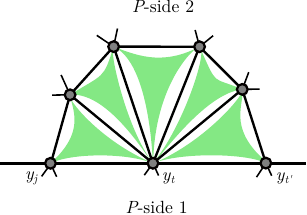}
  \end{center}
  
  \smallskip\noindent In this figure, the green triangles correspond
  to the clique vertices of the web $G+F$. The path of this fan graph
  avoiding $y_t$ is a $y_j$-$y_{t'}$ path $Q$. In what follows we
  assume we have a plane drawing of $G+F$ which induces a plane
  embedding of its rib, and such that clique vertices and associated
  edges of a given triangle of the rib are restricted to the area of
  the plane bounded by the said triangle.

  Let $z$ be an inner vertex of $Q$ on $P$-side $2$ in $G$. We prove
  $zy_t\in F\setminus E(G)$.

  Write $u$ and $v$ for the predecessor and successor of $z$ on $Q$,
  respectively. The edge $y_tz$ is in $F\setminus E(G)$ ($zy_t$ is an
  edge we want to remove, as explained above), since, together with a
  $P$-path from $z$ to $\mathring{P_2}$, having $zy_t\in E(G)$ proves
  that $y_t$ is on $P$-side $2$, which is a contradiction in $G$.

  We describe a process modifying $F$ to remove $z$ from $y_t$'s
  neighbours. Remove $y_tz$ from $F$ as well as the edges between
  clique vertices of the triangles $uzy_t$ and $zy_tv$ of $G+F$'s rib
  connecting a vertex on no $P$-side and one on $P$-side $2$ in $G$. Add
  the edge $uv$ and fill back the cliques in a compatible way with
  respect to the new triangles $uzv$ and $uy_tv$. The whole process is
  represented as follows:
  
  \smallskip
  \begin{center}
    \includegraphics{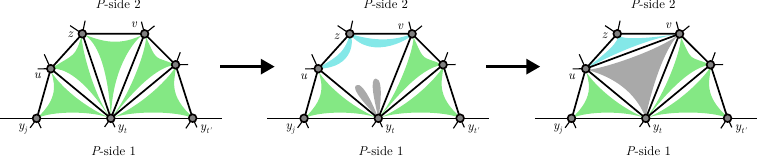}
  \end{center}

  \smallskip\noindent
  In this figure, the grey loops are the clique vertices on no
  $P$-sides in $G$, and the blue thick arcs the ones that have
  $P$-side $2$ in $G$.

  The graph obtained at the end of this process is a web. Its rib is
  induced by the same vertices as the rib of $G+F$. The only possible
  issue would be that $uv$ is already an edge in $G+F$. That cannot
  be as then in the drawing of $G+F$ fixed above, $uvy_t$ would have
  been a triangle of the rib while not being an inner face.

  We distinguish several cases depending on which $P$-sides in $G$ the
  inner vertices of $Q$ are on. In the first case, we are able to
  remove edges in a way close to the illustration provided above, and
  hence we obtain that $y_jy_t'$ is safe. In the second case, more
  edges needs to be removed to prove that $y_jy_{t''}$ is safe for
  $t''\ge t'$.
  \begin{itemize}
  \item \textbf{All the inner vertices of $Q$ are on $P$-side $2$ in
      $G$;} after applying the process above $|Q|-2$ times, we get a
    web containing $G$ as well as the edge $y_jy_{t'}$, proving its
    safeness. In this web the path $P'y_jy_{t'}P'$ is smaller in
    length than $P'$ allowing us to conclude by induction hypothesis.

  \item \textbf{There exists an inner vertex of $Q$ on no $P$-side in $G$;} by
    applying the process above repeatedly, we can assume that all inner
    vertices of $Q$ have no $P$-side in $G$.

    Consider $H$ the maximal (with regard to vertex sets inclusion)
    connected subgraph of $G+F$ induced by the vertices that
    are $C\cup P'$-connected to $y_t$, on $P'$-side $2$ in $G+F$, as
    well as on no $P$-side in $G$. All inner vertices of $Q$ are in
    $H$.

    The first vertex of $P'$ that is adjacent to a vertex of $H$ is
    $y_j$. Indeed, if it was a vertex $y_a$ with $a<j$, then we could
    consider $P$-paths $Q_1$ and $Q_2$ in $G$ from $\mathring{P_1}$ to
    $y_j$ and from $y_j$ to $\mathring{P_2}$ respectively and $R$ a
    $y_a$-$y_t$ path in $H$, and $(Q_1Q_2,Py_aRy_tP)$ would be a
    crossing of $C$ in $G+F$, a contradiction.

    Let $y_{t''}$ be the last vertex of $P'$ that is adjacent to some
    vertex of $H$.

    We prove that $y_jy_{t''}$ is safe by modifying $F$ in order to
    add $y_jy_{t''}$. In $G+F$, we remove all edges between vertices of
    $H$ and vertices of $G+F$ that are not on $P'$. Those are edges
    between vertices having no $P$-side and on $P$-side $2$ in $G$:
    they are all in $F\setminus E(G)$. Next we add the edge
    $y_jy_{t''}$. In the plane drawing of $G+F$ we do that close
    enough to $H$ making it clear that the cycle $y_jP'y_{t''}y_j$
    separates vertices of $H$ from the rest of the graph:
    
    \smallskip
    \begin{center}
      \includegraphics{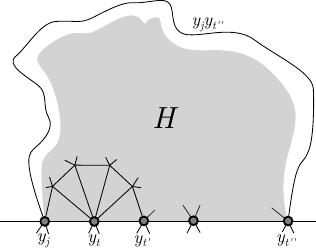}
    \end{center}

    \smallskip\noindent Let $G'$ be the obtained graph. It is clear
    that $G$ is an edge subgraph of $G'$. We prove that $C$ is
    crossless in $G'$. By contradiction, assume that $(R,S)$ is a
    crossing of $C$ in $G'$. One of the paths, say $R$, must contain
    $y_jy_{t''}$. The path $S$ must contain some vertices of
    $y_jP'y_{t''}$ as otherwise $(Ry_jP'y_{t''}R,S)$ would be a
    crossing of $C$ in $G+F$. Let $y_a$ and $y_b$ be the first and
    last inner vertices of $y_jP'y_{t''}$ appearing on $S$. As the
    cycle $y_jP'y_{t''}y_j$ separates $H$ from the rest of $G'$, and
    from $C$ in particular, and as the endpoints of $S$ are on $C$,
    the path $Sy_aP'y_bS$ does not intersect $H$. Since $H$ is connected
    we consider a $y_j$-$y_{t''}$ path $Q'$ such that $\mathring{Q'}$
    is a path in $H$. The pair $(Ry_jQ'y_{t''}R,Sy_aP'y_bS)$ is a
    crossing of $C$ in $G+F$, a contradiction.

    Take a web completion of $G'$ and replace the path $P'$ by
    $P'y_jy_{t''}P'$. The latter is shorter than $P'$, allowing us to
    conclude by induction hypothesis.
  \end{itemize}
\end{proof}

Thanks to this lemma, we can compute $P$-completions using
Algorithm~\ref{algo:P_completion}. Its termination is ensured by the
fact that, in the while loop, the length of $P$ is non-increasing
while the longest prefix $Py$ of $P$ whose vertices are all on both
$P$-sides sees its length increase at each step.
\begin{algorithm}[t]
  \caption{Computing $P$-completions.}
  \label{algo:P_completion}
  \begin{algorithmic}[1]
    \Procedure{\bf{P\_completion}}{$G,C,P$}
      \State $P_1,P_2\gets$ the two subpaths of $C$ separated by $P$ such that $C=P_1\overline{P_2}$
      \State $y\gets$ the first vertex of $P$
      \While{$y$ is not the last vertex of $P$}
        \State $z\gets$ the first vertex of $P$ strictly after $y$ that is on some $P$-side $i$
        \State $u\gets$ the first vertex of $P$ after $z$ on $P$-side $j\neq i$
        \State $G\gets G+yu$
        \State $P\gets PyuP$
        \State $y\gets$ the first vertex of $P$ after $y$ whose successor is not on both $P$-sides
      \EndWhile
      \State \textbf{return} $G,P$
    \EndProcedure
  \end{algorithmic}
\end{algorithm}

\subsection{Extracting crossings and web completions}
\label{sec:computing_crossings_and_web_completions}

In the previous two sections we have shown how to decide the
\textsc{2-Disjoint Paths or Web Completion} problem. Now, we show
how Algorithm~\ref{algo:general_method} can be refined to compute
crossings or web completions explicitly. We use the same notations
$G,C,P,P_i,G_i,C_i$ as above.

When Algorithm~\ref{algo:general_method} returns True we compute a
crossing of $C$ as follows:
\begin{itemize}
\item if the $C$-path $P$ chosen in $G$ does not separate $C$ then we
  compute a $\mathring{P_1}$-$\mathring{P_2}$ path $P'$ in $G-P$ and
  return $(P,P')$.
\item if some recursive call $\mathbf{2DP}(G_i,C_i)$ returns True, and
  hence returns a crossing $(Q,R)$ of $C_i$ in $G_i$, its shape must
  fall into one of the five cases of Figure~\ref{fig:five_cases}. The
  solutions presented there to build a crossing of $C$ in $G$
  still work here, thanks to vertices of $P$ being all on both
  $P$-sides. Returning this crossing then reduces to computing
  appropriate $C\cup P$-paths from $Q$ or $R$ to $\mathring{P_1}$
  and/or $\mathring{P_2}$.
\end{itemize}

In the case Algorithm~\ref{algo:general_method} returns False, we
compute web completions as follows:
\begin{itemize}
\item if there are no $C$-paths between non-consecutive
  vertices of $C$, then each component of $G-C$ is incident to either
  one or two vertices of $C$, and in the latter case the two vertices
  must be consecutive in $C$. We choose $x_0$ in $C$ and set
  $F=\{x_0y\mid y\in V(C),~y\neq x_0\}\cup C$. The graph $R=(V(C),F)$
  is a rib compatible with $G$. For each component of $G-C$ we choose
  a triangle of $R$ it is incident to. For each triangle $T$ of $R$ we
  turn the components of $G-C$ we associated with $T$ into a clique on
  $T$. This provides a web completion of~$G$.
\item if both calls $\mathbf{2DP}(G_i,C_i)$ returned False, and hence
  a web completion of the graph $G_i$ with respect to $C_i$ for $i\in\{1,2\}$,
  then their union is a web completion of $G$ with respect to $C$.
\end{itemize}

\subsection{Complexity} \label{sec:complexity} Let $G$ be a graph with
a cycle $C$ and a $C$-path $P$ between non-consecutive vertices of
$C$. To analyse the complexity of the whole algorithm, we need to be
more precise about:
\begin{itemize}
\item how do we compute paths,
\item how do we decide whether a path separates $C$, and
\item how do we decide $P$-connectedness and in particular which
  $P$-sides a given vertex belongs to.
\end{itemize}

All of these can be done using a \emph{search algorithm} (e.g.
depth-~or breadth-first searches) forbidding the traversal of some
vertices $X\subseteq V(G)$. We call \emph{$X$-searches} such searches.

Whereas a usual search algorithm allows one to compute the connected
components of a graph, $X$-searches allow for the computation of
components of $G-X$ as well as the subsets of $X$ these components are
incident to.

Let $n=|V(G)|$ and $m=|E(G)|$.
\begin{proposition}
  There exists an implementation of
  Algorithm~\ref{algo:general_method} that runs in $O(nm)$ time on
  connected graphs.
\end{proposition}
\begin{proof}
  Let $(G,C)$ be an instance of $\mathbf{2DP}$ such that $G$ is
  connected. An $X$-search can be done in $O(m)$ assuming that the
  relation $-\in X$ is an attribute of vertices of $G$ accessible in
  constant time. Indeed, this is a simple variation of usual search
  algorithms.

  Beside the call to $\mathbf{P\_completion}(G,C,P)$ the complexity of
  one recursive call to Algorithm~\ref{algo:general_method} lies in
  four lines: the two $\mathbf{if}$ statements, the first computation
  of $P$, and the computation of the two subinstances. Each can be
  done in $O(m)$ time by using either a $C$-search or a $P$-search.

  Hence, if $p$ is the number of recursive calls made to
  $\mathbf{2DP}$ and $K$ the complexity of all calls to
  $\mathbf{P\_completion}$, Algorithm~\ref{algo:general_method} runs
  in $O(pm+K)$. Each time a new call to $\mathbf{2DP}$ is made, it
  means an instance has been divided in two subinstances by finding a
  $C$-paths $P$ separating $C$. It is not difficult to see that the
  union $\mathcal{C}$ of all cycles passed as arguments of
  $\mathbf{2DP}$ is a planar graph: following the recursive calls, we
  start with a cycle and we add paths separating inner faces
  recursively:
  
  \smallskip
  \begin{center}
    \includegraphics{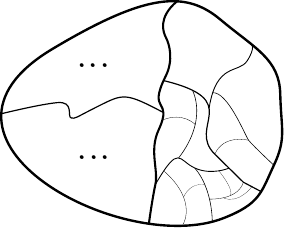}
  \end{center}
  
  \smallskip\noindent
  In the figure, the outer cycle is the first cycle $C$, while inner
  $C$-paths are the recursive separating paths $P$ computed in the
  different calls of $\mathbf{2DP}$. The thinner the line the deeper
  the recursive call.

  The number $p$ of calls to $\mathbf{2DP}$ is bounded by the number
  of edges of this planar graph, meaning $p=O(n)$. Hence the
  complexity of Algorithm~\ref{algo:general_method} is $O(nm+K)$.

  Finally we look at Algorithm~\ref{algo:P_completion}. It requires
  one $P$-search at the beginning to compute $P_1$ and $P_2$, as well
  as one $(C\cup P)$-search for each step of the while loop to compute
  the $P$-sides each vertex is on. The number of such searches among
  all calls to $\mathbf{P\_completion}$ in a run of $\mathbf{2DP}$ is
  bounded by the sum of $p$ with the number of edges added by
  $\mathbf{P\_completion}$. Both are bounded by the number of edges in
  $\mathcal{C}$ which is a $O(n)$. Hence, $K=O(nm)$ and the complexity
  of Algorithm~\ref{algo:general_method} is $O(nm)$.
\end{proof}

To compute crossings or web completions as described in
Section~\ref{sec:computing_crossings_and_web_completions}, the
overall complexity is not modified. We have to be careful about how to
represent web completions, however.

Indeed, imagine that \emph{most} vertices of $G$ are in a component
$D$ of $G-C$ that is incident to exactly one vertex of $C$
($n=|G|=O(|D|)$). If $G-D$ is a rib, then in any web completion of
$G$, all vertices of $D$ are in the clique of some triangle of $G-D$.
This clique contains $\Omega(|D|^2)=\Omega(n^2)$ edges. If $D$ is
sparse, say it is a tree, then $G$ takes only $O(n)$ space in memory
to be stored using adjacency lists while all its web completions have
$\Omega(n^2)$ edges.

To avoid this problem, instead of representing a web using adjacency
lists, we distinguish between rib and clique vertices. The adjacencies
between rib vertices are kept in the form of adjacency lists, while
the adjacencies of clique vertices are held in the form of the three
rib vertices forming their corresponding triangle. As ribs are planar
they have $O(n)$ edges. This provides an $O(n)$-space representation of
webs in memory.

\section*{Acknowledgements}
The authors would like to thank Stéphan Thomassé for fruitful
discussions and for pointing out that an algorithm could hide behind our
proof of the two paths theorem.

\clearpage
\bibliographystyle{siamplain}
\bibliography{references}
\end{document}